\documentclass[letterpaper, 10 pt, conference]{ieeeconf}

\IEEEoverridecommandlockouts                              %
\makeatletter

\let\proof\@undefined
\let\endproof\@undefined
\makeatother

 \usepackage{graphicx}          %
 \usepackage{amssymb}
 \usepackage{amsmath}
\usepackage{amsthm}
 \usepackage{tikz}
 \usepackage{pgfplots}
\let\labelindent\relax
\usepackage{enumitem}

\usepackage{stmaryrd}
\newcommand{\sra}{\shortrightarrow}

\newtheorem{definition}{Definition}
\newtheorem{thm}{Theorem}
\newtheorem{prop}{Proposition}
\newtheorem{assum}{Assumption}
\newtheorem{example}{Example}

\newcommand{\dens}{x}

\newcommand{\fluxin}{s}
\newcommand{\fluxout}{d}

\newcommand{\densjam}{\overline{\dens}}

\newcommand{\network}{\mathcal{G}}
\newcommand{\Verts}{\mathcal{V}}

\newcommand{\Links}{\mathcal{L}}

\newcommand{\Lin}{\mathcal{L}^\textnormal{in}}
\newcommand{\Lout}{\mathcal{L}^\textnormal{out}}
\newcommand{\Lup}{\mathcal{L}^\textnormal{up}}
\newcommand{\Ldown}{\mathcal{L}^\textnormal{down}}
\newcommand{\Ladj}{\mathcal{L}^\textnormal{adj}}

\newcommand{\dense}{\dens^\textnormal{e}}

\newcommand{\Ramps}{\mathcal{R}}

\newcommand{\head}{{\sigma}}
\newcommand{\tail}{{\tau}}

\newcommand{\Domain}{\mathcal{X}}

\newcommand{\Dom}{\mathcal{X}}

\newcommand{\flow}{f}
\newcommand{\alphaF}{\alpha^\textnormal{F}}
\newcommand{\alphaNF}{\alpha^\textnormal{NF}}

\renewcommand{\phi}{\varphi}
\mathcode`l="8000
\begingroup
\makeatletter
\lccode`\~=`\l
\DeclareMathSymbol{\lsb@l}{\mathalpha}{letters}{`l}
\lowercase{\gdef~{\ifnum\the\mathgroup=\m@ne \ell \else \lsb@l \fi}}%
\endgroup

\usetikzlibrary{arrows}
\usetikzlibrary{shapes}
 \usetikzlibrary{calc}
\usetikzlibrary{decorations.pathreplacing}
\usetikzlibrary{arrows,positioning} 
\usetikzlibrary{pgfplots.groupplots}
\usetikzlibrary{decorations.text}

\tikzstyle{link}=[line width=2pt, ->,>=latex]
\tikzstyle{junc}=[draw,circle,inner sep=1pt,minimum width=8pt]
\tikzstyle{onramp}=[line width=2pt, dashed,->,>=latex]

\title{Mixed Monotonicity of Partial First-In-First-Out Traffic Flow Models}
\author{Samuel Coogan, Murat Arcak, Alexander A. Kurzhanskiy\thanks{S. Coogan is with the Electrical Engineering Department at the University of California, Los Angeles, \texttt{scoogan@ucla.edu}. M. Arcak is with the Electrical Engineering and Computer Sciences Department at the University of California, Berkeley, \texttt{arcak@eecs.berkeley.edu}. A. Kurzhanskiy is with California Partners for Advanced Transportation Technology at the University of California, Berkeley, \texttt{akurzhan@berkeley.edu}. This work was supported in part by NSF Grant CNS-1446145.}}

\begin{document}
\maketitle
\begin{abstract}
 In vehicle traffic networks, congestion on one outgoing link of a diverging junction often impedes flow to other outgoing links, a phenomenon known as the first-in-first-out (FIFO) property. Simplified traffic models that do not account for the FIFO property result in monotone dynamics for which powerful analysis techniques exist. FIFO models are in general not monotone, but have been shown to be mixed monotone---a generalization of monotonicity that enables similarly powerful analysis techniques. In this paper, we study traffic flow models for which the FIFO property is only partial, that is, flows at diverging junctions exhibit a combination of FIFO and non-FIFO phenomena.  We show that mixed monotonicity extends to this wider class of models and establish conditions that guarantee convergence to an equilibrium.

\end{abstract}
\section{Introduction}

In models of vehicular traffic flow, if congestion on one outgoing link of a diverging junction impedes the incoming flow to other outgoing links, the diverging junction is said to satisfy the \emph{first-in-first-out (FIFO)} property. If complete congestion on one outgoing link completely blocks access to all other outgoing links, we say the model is a \emph{full FIFO} model. 

Whether a node model of a diverging junction is FIFO or non-FIFO affects the qualitative dynamical behavior of traffic flow through the junction. %
An attractive feature of non-FIFO node models is that the resulting traffic network dynamics are monotone, as is shown in \cite{Lovisari:2014qv}. Trajectories of a \emph{monotone} dynamical system preserve a partial order over the system's state \cite{Hirsch:1985fk,Smith:2008fk}. Preservation of this partial order imposes restrictions on the behavior exhibited by such systems which is exploited for, \emph{e.g.}, characterization of equilibria and stability analysis in \cite{Lovisari:2014qv}.

In general, FIFO node models are not monotone. Nonetheless, in \cite{Coogan:2015mz}, it is shown that a particular full FIFO model is \emph{mixed monotone}, which significantly generalizes the class of monotone systems \cite{Smith:2008rr, Coogan:2014ty}. 

However, non-FIFO models and full FIFO models are often inadequate.  {Non-FIFO} models imply that, even if one of the output links is jammed, the resulting traffic spillback has no effect on those vehicles directed to the other output links, an unreasonable assumption as the FIFO effect in traffic flow networks has been observed even for multilane diverging junctions \cite{Cassidy2002563,Munoz:2002qv}. On the other hand, a full FIFO model is often too restrictive~\cite{wrightetal16a}. To derive a class of {\it partial} FIFO models,~\cite{bliemer07} and \cite{shiomi2015} have suggested modeling lanes of an input link as separate links. The main drawback of this approach is that it greatly complicates the size and dimensionality of the model since every node in the network becomes a multi-input-multi-output junction.%

Here, we propose a general class of \emph{partial FIFO} junction models where the full FIFO rule is relaxed; see Figure \ref{fig:div}. We show that the resulting dynamics are mixed monotone, and we use mixed monotonicity to establish convergence to an equilibrium point of the resulting dynamics. By considering the dynamical properties of FIFO traffic flow models that are not full FIFO, this paper bridges an important gap in the literature.

In Section \ref{sec:network-flow-model}, we present a general model of traffic flow that encompasses many existing non-FIFO and full FIFO models and allows for a large class of partial FIFO models. In Section \ref{sec:mixed-monot-traff}, we show that our general model is mixed monotone. In Section \ref{sec:examples}, we present several specific, practically motivated instantiations of this general model. We show how mixed monotonicity is used for analysis in Section \ref{sec:example} and provide concluding remarks in Section \ref{sec:conclusions}.

\begin{figure}
  \centering
\begin{tabular}{@{}c c@{}}
 \begin{tikzpicture}[scale=.08]
\path[use as bounding box] (15,16.9) rectangle (70,38);
\fill[gray] (15,30) rectangle (38,38);
\fill[gray] (37.8,32) rectangle (70,38);
\fill[gray] (35,30) .. controls +(10,0) and (50,16.9) .. (60,16.9) -- (70,16.9) -- (70,21) -- (62,21) .. controls (52,21) and (48,34) .. (35,34);
\draw[dashed, white] (15,32) -- (36,32) .. controls +(10,0) and (51,19) .. (61,19) -- (70,19);
\draw[dashed, white] (15,34) -- (70,34);
\draw[dashed, white] (15,36) -- (70,36);
  \end{tikzpicture}  

&
  \begin{tikzpicture}[scale=1.5]
    \node[junc] (d) at (-1,0) {};
    \node[junc] (a) at (0,0) {};
    \node[junc] (b) at (30:1) {};
    \node[junc] (c) at (-30:1) {};
    \draw[link] (a) --node[above]{$2$} (b);
    \draw[link] (a) --node[below]{$3$} (c);
    \draw[link] (d) -- node[above, pos=.3]{$1$} (a);
  \end{tikzpicture}
\end{tabular}
\caption{A diverging junction with one incoming link and two outgoing links as on a freeway (left) and schematically (right). When traffic flow is assumed to obey the \emph{first-in-first-out} (FIFO) property, congestion on link 2 (resp. 3) impedes flow to link 3 (resp. 2) whereas in a non-FIFO flow model, the flow from link 1 to link 2 (resp. link 3) is independent of the congestion on link 3 (resp. 2). For the FIFO case, if complete congestion on one outgoing link completely impedes the outgoing flow from link 1, the node model at the diverging junction is said to be \emph{full FIFO}, otherwise it is a \emph{partial FIFO} model.}
\label{fig:div}
\end{figure}
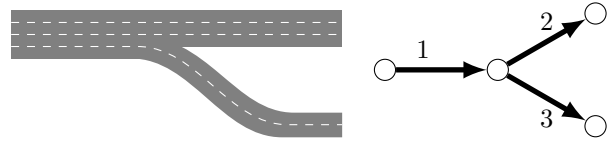

\section{Network Flow Model}
\label{sec:network-flow-model}
A traffic flow network consists of a directed graph $\network=(\Verts,\Links)$ with \emph{junctions} or \emph{nodes} $\Verts$ and \emph{links} $\Links$.   Let $\head(l)$ and $\tail(l)$ denote the head and tail junction of link $l\in\Links$, respectively, where we assume $\head(l)\neq \tail(l)$, \emph{i.e.}, no self-loops. Traffic flows from $\tail(l)$ to $\head(l)$. 

For each $v\in\Verts$, we denote by $\Lin_v\subset \Links$ the set of input links to node $v$ and by $\Lout_v\subset \Links$ the set of output links, \emph{i.e.},  
\begin{align}
  \label{eq:2}
\Lin_v&\triangleq\{l\mid\head(l)=v\},\\
\Lout_v&\triangleq\{l\mid\tail(l)=v\}.
\end{align}
For each $l$, we denote by $\Lup_l\subset\Links$ the set of links immediately upstream of link $l$, and by $\Ldown_l\subset \Links$ the set of links immediately downstream of link $l$. We say that links $l$ and $k$ are \emph{adjacent} if $\tail(l)=\tail(k)$ and $l\neq k$ and let $\Ladj_l\subset \Links$ be the set of links adjacent to link $l$. Thus
\begin{align}
  \label{eq:46}
\Lup_l&\triangleq\{k\in\Links\mid \head(k)=\tail(l)\} = \Lin_{\tail(l)},\\
\Ldown_l&\triangleq\{k\in\Links\mid \tail(k)=\head(l)\}= \Lout_{\head(l)},\\
\Ladj_l&\triangleq\{k\in\Links\mid \tail(k)=\tail(l), k\neq l\} = \Lout_{\tail(l)}\backslash\{l\}.
\end{align}

Each link $l\in\Links$ has state $x_l(t)\geq 0$ evolving over time that is the density of vehicles on link $l$. We denote the  state of the network by $x(t)\triangleq \{x_l(t)\}_{l\in\Links}$. Vehicles flow from link to link over time; the state-dependent flow of vehicles from link $k$ to link $l$ is denoted by $f_{k\sra l}(x)$. We assume $f_{k\sra l}(x)\equiv 0$ if $\head(k)\neq \tail(l)$ so that flow is allowed only between links connected via a junction.  Furthermore, vehicles flow to link $l$ from outside the network at rate $ f_{\sra l}(x)$ and vehicles leave the network from link $l$ at rate $f_{l\sra}(x)$ so that
\begin{align}
  \label{eq:42}
  \dot{x}_l&=\sum_{k\in\Links} f_{k\sra l}(x)-\sum_{j\in\Links}f_{l\sra j}(x) +f_{\sra l}(x) - f_{l\sra }(x)\\
  \label{eq:42-2}&=:F_l(x).
\end{align}

In Section \ref{sec:examples}, we suggest specific forms for $f_{k\sra l}$, $f_{l \sra}$, and $f_{\sra l}$ based on phenomenological properties of traffic flow. %

\newcommand{\fF}{f^\textnormal{F}}
\newcommand{\fNF}{f^\textnormal{NF}}

We further assume that each $f_{k\sra l}(x)$ is decomposable as
\begin{align}
  \label{eq:43}
  f_{k\sra l}(x) =   \fF_{k\sra l}(x)+ \fNF_{k\sra l}(x),
\end{align}
where $\fF_{k\sra l}(x)$ is the flow from link $k$ to link $l$ that is subject to the FIFO phenomenon and $\fNF_{k\sra l}(x)$ is the flow from link $k$ to link $l$ that is not subject to the FIFO phenomenon.

The following captures the fundamental properties of traffic flow networks.
\begin{assum}
\label{assum:main}
For all $l,k\in\Links$, the functions $f_{k\sra l}(x)$, $f_{l\sra}(x)$, $f_{\sra l}(x)$ are locally Lipschitz continuous. For all $x\in\Dom$ where the given derivative exists,

\noindent \textbf{External flows:}
  \begin{enumerate}[label={(A\arabic*)},labelindent=*,leftmargin=*]
\item\label{eq:44}
 $\displaystyle\frac{\partial f_{\sra l}}{\partial x_m}(x) \geq 0$ for all $ l,m\in\Links \text{ such that } m\neq l$,\\
\item \label{eq:44-2}   $\displaystyle\frac{\partial f_{l \sra}}{\partial x_m}(x) \leq 0$ for all $ l,m\in\Links \text{ such that } m\neq l$.
  \end{enumerate}
Interpretation:
  \begin{itemize}[leftmargin=*]
\item \ref{eq:44}: For any $m\neq l$, increasing the density on link $m$ can only increase the exogenous flow into link $l$. For example, congestion on link $m$ of the network causes vehicles that wish to enter the network to reroute and enter at link $l$.
\item \ref{eq:44-2}: For any $m\neq l$, increasing the density on link $m$ can only decrease the flow that exits the network from link $l$. For example, downstream congestion on link $m$ impedes the outflow of vehicles via an offramp on link $l$. 
  \end{itemize}
\noindent \textbf{Local dependence:}
\begin{enumerate}[label={(A\arabic*)},labelindent=*,leftmargin=*]
\setcounter{enumi}{2}
\item \label{eq:44-3}  $\displaystyle\frac{\partial \fNF_{k \sra l}}{\partial x_m}(x) \equiv 0$ for all $l,k,m\in\Links$ such that \\$m\not\in\Lin_{\head(k)}\cup\Lout_{\head(k)}$,
\item \label{eq:44-4}  $\displaystyle\frac{\partial \fF_{k \sra l}}{\partial x_m}(x) \equiv 0$ for all $l,k,m\in\Links$ such that  \\ $m\not\in\Lin_{\head(k)}\cup\Lout_{\head(k)}$.
  \end{enumerate}
Interpretation:
\begin{itemize}[leftmargin=*]
\item \ref{eq:44-3} and \ref{eq:44-4}: The flow rate from link $k$ to link $l$ through some junction $v=\head(k)=\tail(l)\in\Verts$ is instantaneously affected by the change in density of vehicles on link $m$ only if $m$ is incoming or outgoing of junction $v$.
\end{itemize}

\noindent \textbf{Net incoming and outgoing flows:}
\begin{enumerate}[label={(A\arabic*)},labelindent=*,leftmargin=*,itemsep=2pt]
\setcounter{enumi}{4}
\item \label{eq:44-7} $\displaystyle\frac{\partial }{\partial x_m}\Bigg(\sum_{j\in\Links}\fF_{j\sra l}\Bigg) (x)\geq 0$ for all $l,m\in\Links$ such that $m\in\Lup_l$,
\item \label{eq:44-72} $\displaystyle\frac{\partial }{\partial x_m}\Bigg(\sum_{j\in\Links}\fNF_{j\sra l}\Bigg) (x)\geq 0$ for all $l,m\in\Links$ such that $m\in\Lup_l$,
\item \label{eq:44-8}  $\displaystyle\frac{\partial }{\partial x_m}\Bigg(\sum_{j\in\Links}f_{l\sra j}\Bigg) (x)\leq 0$ for all $l,m\in\Links$ such that $m\in \Lin_{\head(l)}\cup\Lout_{\head(l)}, \ m\neq l$.
  \end{enumerate}
Interpretation:
\begin{itemize}[leftmargin=*]
\item \ref{eq:44-7} and \ref{eq:44-72}: For any link $m$ immediately upstream of link $l$ (that is, $\head(m)=\tail(l)$), increasing the density of vehicles on link $m$ cannot decrease the net incoming FIFO or non-FIFO flow to link $l$.
\item \ref{eq:44-8}: For any link $m\neq l$ incoming or outgoing from junction $\head(l)$, increasing the density of vehicles on link $m$ cannot increase the net outgoing flow from link $l$.
\end{itemize}
\noindent \textbf{FIFO and non-FIFO flows:}
\begin{enumerate}[label={(A\arabic*)},labelindent=*,leftmargin=*]
\setcounter{enumi}{7}
\item \label{eq:44-5} $\displaystyle\frac{\partial \fNF_{k \sra l}}{\partial x_m}(x) \geq 0$ for all $ l,k,m\in\Links \text{ such that } m\in\Ladj_l$,
\item \label{eq:44-6} $\displaystyle\frac{\partial \fF_{k \sra l}}{\partial x_m}(x) \leq 0$, for all $l,k,m\in\Links \text{ such that } m\in\Ladj_l$.
  \end{enumerate}
Interpretation:
\begin{itemize}[leftmargin=*]
\item \ref{eq:44-5}: For any link $m$ adjacent to link $l$, increasing the density of link $m$ can only increase the non-FIFO flow from an upstream link $k$ to $l$. This may occur if, \emph{e.g.}, vehicles reroute to avoid increased congestion on link $m$.
\item \ref{eq:44-6}: For any link $m$ adjacent to link $l$, increasing the density of link $m$ can only decrease the FIFO flow from an upstream link $k$ to link $l$. This captures the fundamental feature of FIFO flow whereby congestion on link $m$ blocks access to link $l$.
\end{itemize}
\end{assum}

Requirements \ref{eq:44}--\ref{eq:44-8} are standard for traffic flow networks. The requirement \ref{eq:44-5} is found in, \emph{e.g.}, \cite[Definition 2]{Como:2015ne} where it is used to establish monotonicity for non-FIFO policies. Requirement \ref{eq:44-6} captures the FIFO phenomenon.

\section{Mixed Monotonicity of Traffic Flow}
\label{sec:mixed-monot-traff}
\begin{definition}[Mixed Monotone]
\label{def:cwmm}
The system $\dot{x}=G(x)$, $x\in X\subseteq \mathbb{R}^n$ where $X$ has convex interior and $G$ is locally Lipschitz is \emph{mixed monotone} if there exists a locally Lipschitz continuous function $g(x,y)$ satisfying:
\begin{enumerate}
\item $g(x,x)=G(x)$ for all $x\in X$,\\
\item $\displaystyle \frac{\partial g_i}{\partial x_j} (x,y)\geq 0$ for all $x,y\in X$ and all  $i\neq j$ whenever the derivative exists,\\
\item $\displaystyle \frac{\partial g_i}{\partial y_j} (x,y)\leq 0$ for all $x,y\in X$ and all  $i,  j$ whenever the derivative exists.
\end{enumerate}
The function $g(x,y)$ is called a \emph{decomposition function} for the system.
\end{definition}

Let $\dot{x}=G(x)$ be mixed monotone with decomposition function $g(x,y)$ and consider the dynamical system
\begin{align}
  \label{eq:9}
  \dot{x}&=g(x,y),\\
  \label{eq:9-2} \dot{y}&=g(y,x).
\end{align}
The symmetry implies that if $(x(t),y(t))$ is a trajectory of \eqref{eq:9}--\eqref{eq:9-2}, then $(y(t),x(t))$ is also a trajectory.  Furthermore, observe that $\{(x,y)\mid x=y\}$ is an invariant subspace of \eqref{eq:9}--\eqref{eq:9-2} and trajectories contained within this subspace correspond to (two copies of) trajectories of the original system $\dot{x}=G(x)$, thus we refer to \eqref{eq:9}--\eqref{eq:9-2} as the \emph{embedding} system.

The importance of mixed monotonicity lies in the observation that the induced embedding system is monotone with respect to the partial order induced by the orthant $\mathbb{R}_{\geq 0}^n\times \mathbb{R}_{\leq 0}^n$, that is, \eqref{eq:9}--\eqref{eq:9-2}  is monotone with respect to the partial order $\preceq$ defined by $(x,y)\preceq (v,w)$ if and only if $x\leq v$ and $w\leq y$. This observation allows the powerful tools available for monotone systems to be applied to mixed monotone systems. For example, global convergence for \eqref{eq:9}--\eqref{eq:9-2} implies global convergence of the mixed monotone system $\dot{x}=G(x)$. In Section \ref{sec:example}, we apply this technique to prove asymptotic convergence of the example in Figure \ref{fig:div}.

\begin{thm}
\label{thm:MM}
The traffic flow network model \eqref{eq:42} satisfying Assumption \ref{assum:main} is mixed monotone.
\end{thm}
\begin{proof}
  We construct an appropriate decomposition function $g(x,y)$. For each $l\in\Links$, let 
  $  z^l(x,y): \Dom\times\Dom\to\Dom$
be defined elementwise as
\begin{align}
  \label{eq:56}
  z^l_k(x,y)=
  \begin{cases}
    y_k&\text{if }k\in\Ladj_l\\
    x_k&\text{else}
  \end{cases}
\quad \text{for all }l\in\Links.
\end{align}
Define 
\begin{align}
\nonumber  g_l(x,y)&=\sum_{k\in\Links} \left(\fF_{k\sra l}(z^l(x,y))+\fNF_{k\sra l}(x)\right)\\
  \label{eq:58}&-\sum_{j\in\Links}\left(\fF_{l\sra j}(x)+\fNF_{l\sra j}(x)\right)+f_{\sra l}(x) - f_{l\sra }(x)
\end{align}
and let $g(x,y)=\{g_l(x,y)\}_{l\in\Links}$. Then $g_l(x,x)=F_l(x)$ given in \eqref{eq:42}--\eqref{eq:42-2} for all $l\in\Links$.
We next show
\begin{align}
  \label{eq:59}
\frac{\partial g_l}{\partial x_m}(x,y)\geq 0 \text{ for all } m\neq l.
\end{align}
To this end, we show
\begin{align}
  \label{eq:60}
&  \frac{\partial}{\partial x_m}\left(\sum_{k\in\Links} \left(\fF_{k\sra l}(z^l(x,y))+\fNF_{k\sra l}(x)\right)\right) \geq 0\quad \forall m\neq l,\\
&  \label{eq:60-2}  \frac{\partial}{\partial x_m}\left(\sum_{j\in\Links}\left(\fF_{l\sra j}(x)+\fNF_{l\sra j}(x)\right)\right) \leq 0\quad \forall m\neq l,
\end{align}
which, combined with \ref{eq:44} and \ref{eq:44-2} of Assumption \ref{assum:main}, proves \eqref{eq:59}.  We have that \eqref{eq:60-2} holds for all $m\in\Lin_{\head(l)}\cup\Lout_{\head(l)}$ with $m\neq l$ by \ref{eq:44-8}, and \ref{eq:44-3}--\ref{eq:44-4} ensures that \eqref{eq:60-2} holds with equality for all $m\not\in\Lin_{\head(l)}\cup\Lout_{\head(l)}$. For $m\in\Lup_l$, \eqref{eq:60} holds from \ref{eq:44-7} and \ref{eq:44-72}. For $m\in\Ladj_l$, we have $\frac{\partial}{\partial x_m}(\fF_{k\to\l}(z^l(x,y))=0$ for all $k$ by \eqref{eq:56}, and $\frac{\partial}{\partial x_m}(\fNF_{k\to\l}(x))\geq 0$ by \ref{eq:44-5}, satisfying \eqref{eq:60}. For $m\not\in \Ladj_l\cup\Lup_l$, we have \eqref{eq:60} holds with equality by \ref{eq:44-3}--\ref{eq:44-4}.

We now show
\begin{align}
  \label{eq:61}
  \frac{\partial g_l}{\partial y_m}(x,y)\leq 0 \text{ for all } m\neq l.
\end{align}
We have that \eqref{eq:61} holds trivially for all $m\not\in\Ladj_l$ by \eqref{eq:56}. For $m\in\Ladj_l$, we have
\begin{align}
  \label{eq:1}
    \frac{\partial g_l}{\partial y_m}(x,y) &= \frac{\partial}{\partial y_m}\left(\sum_{j\in\Links} \fF_{j\sra l}(z^l(x,y))\right)\leq 0,
\end{align}
where the inequality follows by \ref{eq:44-6}.
\end{proof}

We remark that a sufficient condition for mixed monotonicity of $\dot{x}=G(x)$ is for each off-diagonal entry of the Jacobian matrix $\frac{\partial G}{\partial x}$ to not change sign over the domain $\Dom$, that is, either $\frac{\partial G_i}{\partial x_j}(x)\geq 0$ for all $x\in\Dom$  or $\frac{\partial G_i}{\partial x_j}(x)\leq 0$ for all $x\in\Dom$ for all $i\neq j$. This condition is proved for the discrete-time case in \cite{Coogan:2014ty} and the proof for the continuous-time case is similar. In general, partial FIFO models do not satisfy this condition; this is attributable to the different sign conditions in \ref{eq:44-5} and \ref{eq:44-6} whereby an increase on some link $k\in\Ladj_l$ may increase the non-FIFO flow to link $l$ and decrease the FIFO flow to link $l$. Thus we require a different construction for the decomposition function as shown in the proof of Theorem \ref{thm:MM}. 

We further observe that standard monotonicity is often generalized to partial orders induced by arbitrary orthants of $\mathbb{R}^n$ \cite{Angeli:2003fv}, and one may wonder if mixed monotonicity for traffic networks is equivalent to this generalization. As observed in \cite{Coogan:2015mz}, it is not difficult to construct traffic network topologies that are not monotone with respect to any orthant order, thus mixed monotonicity is strictly more general.

\section{Examples Of Models Satisfying Assumption \ref{assum:main}}
\label{sec:examples}
We now present several related examples satisfying \ref{eq:44}--\ref{eq:44-6} based on the \emph{supply} and \emph{demand} concept of traffic flow. We assume each link $l\in\Links$ possesses a \emph{jam density} $\densjam_l$ such that $\dens_l(t)\in[0,\densjam_l]$ for all time and thus $\Dom=\prod_{l\in\Links}[0,\densjam_l]$. We further assume each link possesses a state-dependent \emph{demand} function $\fluxout_l(x_l)$ and a state-dependent \emph{supply} function $\fluxin_l(x_l)$ satisfying:
\begin{assum}
\label{assum:2}
  For each $l\in\Links$:
  \begin{itemize}
  \item The demand function $\fluxout_l(\dens_l):[0,\densjam_l]\to \mathbb{R}_{\geq 0}$ is strictly increasing and Lipschitz continuous with $\fluxout_l(0)=0$.
  \item The supply function $\fluxin_l(\dens_l):[0,\densjam_l]\to \mathbb{R}_{\geq 0}$ is strictly decreasing and Lipschitz continuous with $\fluxin_l(\densjam_l)=0$.
  \end{itemize}
\end{assum}

The demand of a link is interpreted as the maximum outflow of the link, and the supply of a link is interpreted as the maximum inflow of the link. %

Let $\Ramps=\{l\in\Links \mid \Lin_{\tail(l)}=\emptyset\}$, that is, $\Ramps$ is the set of links for which there are no upstream links. We assume exogenous traffic enters the network only through $\Ramps$ so that
\begin{align}
  \label{eq:39}
f_{\sra l}(x)\equiv 0\quad \text{for all }l\not\in\Ramps.
\end{align}

For each $l\in\Ramps$, we assume there exists a constant exogenous inflow demand $\delta_l$ such that
\begin{align}
  \label{eq:40}
  \flow_{\sra l}(x) = \min\{\delta_l,\fluxin_l(x_l)\}\text{ for all }l\in\Ramps.
\end{align}

We further assume that $\flow_{l\sra}(x)$ is a fixed fraction $\gamma_l$ of the total outflow from link $l$ if there are any links downstream of $l$, otherwise $\flow_{l\sra}(x)$ is equal to the demand of link $l$. That is, for all $l\in\Links$,
\begin{align}
  \label{eq:41}
  \flow_{l\sra}(x)=\begin{cases}\gamma_l \sum_{j\in\Links}\flow_{l\sra j}(x)& \text{ if } \Lout_{\head(l)}\neq \emptyset\\
\fluxout_l(x_l)&\text{otherwise},
\end{cases}
\end{align}
where $\gamma_l\geq 0$ for each $l$ such that $\Lout_{\head(l)}\neq \emptyset$.

Finally, we assume there exist fixed turn ratios $\beta_{l}> 0$ for each $l$ with $\Lup_l\neq \emptyset$ that describe how vehicles route through the network. The role of these turn ratios is made explicit subsequently, but the interpretation is that $\beta_l$ is the fraction of the upstream demand that is bound for link $l$.

 It remains to characterize $\flow_{k\sra l}(x)$ for all $l,k\in\Links$.

\begin{example}[non-FIFO]
\label{ex:1}
For all $l\in\Links$, let
  \begin{align}
    \label{eq:45}
    \alphaNF_l(x)=\min\left\{1,\frac{\fluxin_l(x_l)}{\beta_{l}\sum_{j\in\Lup_l} \fluxout_j(x_j)}\right\}.
  \end{align}
Let $\fF_{k\sra l}(x)\equiv 0$ for all $k,l\in\Links$ and let
\begin{align}
  \label{eq:47}
  \fNF_{k\sra l}(x)=\alphaNF_l(x) \beta_l\fluxout_k(x_k)\quad \forall l\in\Links, \ \forall k\in\Lup_l.
\end{align}
\end{example}

\begin{example}[Full FIFO]
\label{ex:2}
For all $v\in\Verts$, let
  \begin{align}
    \label{eq:45-2}
    \alphaF_v(x)=\min\left\{1,\min_{k\in\Lout_v}\left\{\frac{\fluxin_k(x_k)}{\beta_{k}\sum_{j\in\Lin_v} \fluxout_j(x_j)}\right\}\right\}.
  \end{align}
Let $\fNF_{k\sra l}(x)\equiv 0$ for all $k,l\in\Links$ and let
\begin{align}
  \label{eq:47}
  \fF_{k\sra l}(x)=\alphaF_{\tail(l)}(x) \beta_l\fluxout_k(x_k)\quad \forall l\in\Links, \ \forall k\in\Lup_l.
\end{align}
\end{example}

\begin{example}[Convex combination of non-FIFO and full FIFO]
\label{ex:3}
  Let $\alphaNF_l(x)$ be given as in \eqref{eq:45} and let $\alphaF_v(x)$ be given as in \eqref{eq:45-2}. Suppose there exists $\eta_l\in[0,1]$ for all $l\in\Links$, and let
  \begin{align}
    \label{eq:48}
  \fF_{k\sra l}(x)&=\eta_l\alphaF_{\tail(l)}(x) \beta_l\fluxout_k(x_k)\quad \forall l\in\Links, \ \forall k\in\Lup_l,\\
      \fNF_{k\sra l}(x)&=(1-\eta_l) \alphaNF_l(x) \beta_l\fluxout_k(x_k)\quad \forall l\in\Links, \ \forall k\in\Lup_l.
  \end{align}
\end{example}

Example \ref{ex:3} is proposed in  \cite[Example 4]{Lovisari:2014qv} and is a natural extension of the ideas in Examples \ref{ex:1} and \ref{ex:2}, however it exhibits the following property: it is possible for $\alphaNF_l(x)<1$ yet $\sum_{j\in\Links}\flow_{j\sra l}(x)<\fluxin_l(x_l)$, that is, the supply of link $l$ restricts the flow to link $l$, yet the total inflow of link $l$ is less than this supply. This property may be undesirable, depending on the specific phenomena which the node model is desired to capture. We now suggest an alternative partial FIFO model. To fix ideas, we assume that each diverging junction has exactly one incoming link, that is,
\begin{align}
  \label{eq:6}
 |\Lout_v|>1\implies |\Lin_v|= 1 \quad \forall v\in\Verts. 
\end{align}
This is not too restrictive as we can model a general diverging junction as a merging node and a node satisfying \eqref{eq:6}.%

\newcommand{\thetaF}{\theta^\textnormal{F}}
\begin{example}[Shared and exclusive lanes for a partial FIFO model]
\label{ex:4}
Assume \eqref{eq:6} holds. We consider $\eta_l\in[0,1]$ for each $l\in\Links$ representing the degree of influence on link $l$ of the FIFO restriction at the intersection so that $\eta_l$ is the fraction of traffic bound for link $l$ that is subject to a FIFO restriction and $(1-\eta_l)$ is the fraction of traffic bound for link $l$ that is not subject to a FIFO restriction. For example, $1-\eta_l$ is the fraction of lanes at the diverging junction exclusively bound for link $l$ and $\eta_l$ is the fraction of lanes that are shared among all outgoing links. 

Whenever $\Ladj_l=\emptyset$, we assume $\eta_l=1$ without loss of generality.
Let $\alphaF_v(x)$ be given as in \eqref{eq:45-2} for all $v\in\Verts$.
For $l\in\Links$ such that $\Ladj_l\neq \emptyset$, let $k$ be the unique link such that $\Lup_l=\{k\}$ (uniqueness is guaranteed by \eqref{eq:6}), and let
\begin{align}
  \label{eq:50}
  \fF_{k\sra l}(x)&=\eta_l\alphaF_{\tail(l)}(x) \beta_l\fluxout_k(x_k),\\
  \label{eq:50-2}  \fNF_{k\sra l}(x)&=\min\left\{(1-\eta_l)\beta_l\fluxout_k(x_k),\fluxin_l(x_l) -  \fF_{k\sra l}(x)\right\}.
\end{align}
For $l\in\Links$ such that $\Ladj_l=\emptyset$, we have that $\eta_l=1$ so that
\begin{align}
  \label{eq:11}
\fF_{k\sra l}(x)&=\alphaF_{\tail(l)}(x) \beta_l\fluxout_k(x_k) \quad \forall k\in\Lup_l,\\
\fNF_{k\sra l}(x)&\equiv 0.
\end{align}
\end{example}

We extend Example \ref{ex:4} to the case where there are multiple sets of interacting outgoing links that result in a collection of FIFO restrictions.
\newcommand{\F}{\mathcal{F}}
\begin{example}
  \label{ex:5}
Assume \eqref{eq:6} holds. For each $v\in\Verts$ with $\Lin_v\neq\emptyset$, let $\Phi(v)\subset 2^{\Lin_v}$ be a collection of subsets of $\Lin_v$ so that each $\varphi\in\Phi(v)$, $\phi\subseteq \Lin_v$ is a set links which are mutually governed by a FIFO restriction. When $|\Lout_v|=1$, we assume $\Phi(v)=\{\Lout_v\}$.

 For $\phi\in\Phi(v)$ and $l\in\Lout_v$, let $\eta_{l,\phi}\in[0,1]$ represent the degree of influence on link $l$ of the FIFO restriction set $\phi$. We make the following assumptions:
\begin{align}
  \label{eq:4-0}  l\not\in \phi&\implies \eta_{l,\phi}=0\quad \forall \phi\in\Phi(\tail(l)),\\
  \label{eq:4}  \sum_{\phi\in\Phi(v)}\eta_{l,\phi}&\leq 1\quad \forall l\in\Lin_v \ \forall v\in\Verts.
\end{align}
Define
 $ \bar{\eta}_l=1-  \sum_{\phi\in\Phi(v)}\eta_{l,\phi}.$
For all $v\in\Verts$ such that $|\Lout_v|>1$, define
\begin{align}
  \label{eq:3}
  \alpha_\phi(x)=\min\left\{1,\min_{j\in\phi}\left\{\frac{\fluxin_j(x_j)}{\beta_{j}\fluxout_k(x_k)}\right\}\right\} \quad \forall \phi\in\Phi(v)
\end{align}
where $k$ is the unique upstream link such that $\Lin_v=\{k\}$.

For $l\in\Links$ such that $\Ladj_l\neq \emptyset$, let $k$ be the unique link such that $\Lup_l=\{k\}$, and let
\begin{align}
  \label{eq:8}
   f^\phi_{k\sra l}(x)&=\eta_{l,\phi}  \alpha_\phi(x) \beta_l\fluxout_k(x_k)\quad \forall \phi\in\Phi(\tail(l)), \\
\fF_{k\sra l}(x)&=\sum_{\phi\in\Phi(\tail(l))}f^\phi_{k\sra l}(x),\\
  \label{eq:8-2}  \fNF_{k\sra l}(x)&=\min\left\{\bar{\eta}_l\beta_l\fluxout_k(x_k),\fluxin_l(x_l) -  \fF_{k\sra l}(x)\right\}.
\end{align}
For $l\in\Links$ such that $\Ladj_l=\emptyset$, we again let
\begin{align}
  \label{eq:5}
\fF_{k\sra l}(x)&=\alphaF_{\tail(l)}(x) \beta_l\fluxout_k(x_k) \quad \forall k\in\Lup_l,\\
\fNF_{k\sra l}(x)&\equiv 0,
\end{align}
where $\alphaF_{\tail(l)}(x)$ is as given in \eqref{eq:45-2}.
\end{example}
Taking $\Phi(v)=\{\Lout_v\}$ for all $v\in\Verts$, we see that Example \ref{ex:4} is a special case of Example \ref{ex:5}.

All the examples above satisfy, for all $x\in\Domain$,
\begin{alignat}{2}
  \label{eq:51}
  \flow_{k\sra l}(x)&\leq \beta_l\fluxout_k(x_k)&&\quad \forall k\in\Links \ \forall l\in\Ldown_k,\\
  \sum_{k\in\Links}\flow_{k\sra l}(x)&\leq \fluxin_l(x_l)&&\quad  \forall l\in\Links.
\end{alignat}

If we further assume that $\sum_{k\in\Ldown_l}\beta_k\leq 1$ and 
$(\gamma_l+1){\sum_{k\in\Ldown_l}\beta_k}\leq 1$
for all $l$ such that $\Lout_{\head(l)}\neq \emptyset$, we have that, for all $x\in\Dom$,
\begin{align}
  \label{eq:52}
  \sum_{k\in\Links}\flow_{l\sra k}(x)+\flow_{l\sra}(x)\leq d_l(x_l) \quad \forall  l\in\Links.
\end{align}

\begin{prop}
\label{sec:examples-1}
  Examples \ref{ex:1}--\ref{ex:5} satisfy Assumption \ref{assum:main}.
\end{prop}
\begin{proof}

It follows straightforwardly from results in \cite{Lovisari:2014qv} that the conditions of Assumption \ref{assum:main} hold for Example \ref{ex:1}, and, similarly, it follows from results in \cite{Coogan:2015mz} that the assumption holds for Example \ref{ex:2}. From Example 1 and Example 2, the assumption immediately holds for Example \ref{ex:3}. We now show that Example \ref{ex:5} satisfies Assumption \ref{assum:main}, from which it follows that also Example \ref{ex:4} satisfies Assumption \ref{assum:main} because Example \ref{ex:4} is a special case of Example \ref{ex:5}. To that end, we now prove each condition \ref{eq:44}--\ref{eq:44-6} for Example 5.

\begin{itemize}[leftmargin=*]
\item    (Condition \ref{eq:44}).  Follows trivially from \eqref{eq:39} and \eqref{eq:40}.
    \item    (Condition \ref{eq:44-2}). Follows from \eqref{eq:41} and Condition \ref{eq:44-8}, proved below, as well as Conditions \ref{eq:44-3} and \ref{eq:44-4}, proved below.
    \item    (Conditions \ref{eq:44-3} and \ref{eq:44-4}). Follows immediately from the fact that for all $v\in\Verts$ and all $\phi\in\Phi(v)$, $\alphaF_v(x)$ and $\alpha_\phi(x)$ are only functions of $d_l(x_l)$ for $l\in\Lin_v$ and $s_l(x_l)$ for $l\in\Lout_v$ and from \eqref{eq:8-2}.
 \item    (Conditions \ref{eq:44-7} and \ref{eq:44-72}). 
Consider $l\in\Links$. If $|\Lup_l|>1$, then $\Ladj_l=\emptyset$ by \eqref{eq:6} and
\begin{align}
  \label{eq:19}
  \sum_{j\in\Links}\fF_{j\sra l}(x)&\in\left\{\sum_{j\in\Lup_l}{\beta_l}\fluxout_j(x_j),\fluxin_l(x_l)\right\},\\
  \sum_{j\in\Links}\fNF_{j\sra l}(x) &\equiv 0,
\end{align}
and thus \ref{eq:44-7} and \ref{eq:44-72} hold. 

Now suppose $|\Lup_l|=1$ and let $\Lup_l=\{k\}$ so that
\begin{align}
  \label{eq:20}
  \sum_{j\in\Links}\fF_{j\sra l}(x)&=\fF_{k\sra l}(x)=\sum_{\phi\in\Phi(\tail(l))}\flow^\phi_{k\sra l}(x),\\
  \sum_{j\in\Links}\fNF_{j\sra l}(x)&=\fNF_{k\sra l}(x).
\end{align}
We have 
\begin{align}
  \label{eq:23}
     \alpha_\phi(x) \fluxout_k(x_k)=\min\left\{\fluxout_k(x_k),\min_{j\in\phi}\left\{\frac{s_j(x_j)}{\beta_j}\right\}\right\},
\end{align}
and thus \ref{eq:44-7} holds by \eqref{eq:8} and \eqref{eq:20}.

Still supposing $|\Lup_l|=1$ with $\Lup_l=\{k\}$, consider now Condition \ref{eq:44-72}. The only possibility for which this condition would not hold is if $  s_l(x_l)-\fF_{k\sra l}(x)$ is the minimizer in \eqref{eq:8-2} and $  \frac{\partial}{\partial x_k}\fF_{k\sra l}(x)>0$. But
  $\frac{\partial}{\partial x_k}\fF_{k\sra l}(x)>0$
only if $\alpha_\phi(x)=1$ for some $\phi$ for which $l\in\phi$ on some neighborhood of $x$ so that, in particular, $s_l(x_l)>\beta_ld_k(x_k)$. In this case, since $\fF_{k\sra l}(x)\leq \sum_{\phi\in\Phi(\tail(l))}\eta_{l,\phi}\beta_l\fluxout_k(x_k)$, we have that
\begin{align}
  \label{eq:22}
  s_l(x_l)-\fF_{k\sra l}(x)\geq \bar{\eta}_l\beta_l\fluxout_k(x_k),
\end{align}
\emph{i.e.}, $\bar{\eta}_l\beta_l\fluxout_k(x_k)$ is the minimizer in \eqref{eq:8-2} and thus \ref{eq:44-72} holds.

    \item    (Condition \ref{eq:44-8}). 
 Suppose $v\in\Verts$ is such that $|\Lout_v|=1$ and let $\Lout_v=\{k\}$. Consider $l\in\Lin_v$. Then
\begin{align}
  \label{eq:12}
&  \sum_{j\in\Links}\flow_{l\sra j}(x) = \fF_{l\sra k} = \alphaF_v(x)\beta_kd_l(x_l)\\
&=\min\left\{\beta_kd_l(x_l),\frac{d_l(x_l)}{\sum_{j\in\Lin_v}d_j(x_j)}s_k(x_k)\right\}.
\end{align}
Since $\frac{\partial}{\partial x_m}\frac{d_l(x_l)}{\sum_{j\in\Lin_v}d_j(x_j)}\leq 0$ for all $m\in\Lin_v$ with $m\neq l$, and $\frac{ds_k}{dx_k}(x_k)\leq 0$, \ref{eq:44-8} holds.

 Now suppose $|\Lout_v|>1$ so that $|\Lin_v|=1$, and let $\Lin_v=\{l\}$. 
From \eqref{eq:8}--\eqref{eq:8-2}, for all $j\in\Lout_v$ we have
\begin{align}
  \label{eq:17}
  \flow_{l\sra j}(x)&=\fF_{l\sra j}(x)+ \fNF_{l\sra j}(x)\\
&\in\left\{\fF_{l\sra j}(x)+\bar{\eta}_j\beta_j\fluxout_l(x_l),s_j(x_j)\right\}.
\end{align}
Consider some $\phi\in\Phi(v)$. Then 
\begin{align}
  \label{eq:16}
   \alpha_\phi(x) \fluxout_l(x_l)=\min\left\{\fluxout_l(x_l),\min_{i\in\phi}\left\{\frac{s_i(x_i)}{\beta_i}\right\}\right\},
\end{align}
so that $\frac{\partial}{\partial x_m}(   \alpha_\phi(x) \fluxout_l(x_l))\leq 0$ for all $m\in\Lout_v$, and
\begin{align}
  \label{eq:15}
  \frac{\partial \fF_{l\sra j}}{\partial x_m}(x)=  \frac{\partial}{\partial x_m}\left(\sum_{\phi\in\Phi(v)}\eta_{j,\phi}  \alpha_\phi(x) \beta_j\fluxout_l(x_l)\right)\leq 0
\end{align}
for all $m\in\Lout_v$. From \eqref{eq:17}, we have that $  \frac{\partial\flow_{l\sra j}}{\partial x_m}(x)\leq 0$ for all $m\in\Lout_v$ so that \ref{eq:44-8} holds. Finally, when $\Lout_{v}=\emptyset$, condition \ref{eq:44-8} holds trivially for all $l\in\Lin_v$.
   
    \item    (Condition \ref{eq:44-5}). Follows from Condition \ref{eq:44-6} below and \eqref{eq:8-2}. 
    \item    (Condition \ref{eq:44-6}).  Follows from \eqref{eq:8}, \eqref{eq:16}, and the fact that $\frac{d s_i}{d x_i}(x_i)\leq 0$ for all $i\in\Links$.
\end{itemize}
\end{proof}

\section{Analysis from Mixed Monotonicity: Diverging Junction}
\label{sec:example}

We return to the example of Figure \ref{fig:div} and suppose the diverging junction obeys the dynamics given in Example \ref{ex:4} of Section \ref{sec:examples}. We assume for all $l\in\{1,2,3\}$
\begin{align}
  \label{eq:7}
  d_l(x_l)&=a_l(1-\exp(-0.5x_l)),\\
  \label{eq:7-2} s_l(x_l)&=b_l-x_l
\end{align}
where $(a_1,a_2,a_3)=(4,3,2)$ corresponds to the lane configuration in Figure \ref{fig:div}, $(b_1,b_2,b_3)=(6,4,2)$, turn ratios $(\beta_{2},\beta_3)=(0.8,0.2)$ indicates that most of the traffic is bound for link 2, and $(\eta_2,\eta_3)=(0.1,0.9)$ indicates that link 3 is strongly affected by the FIFO restriction but link 2 is not. We assume $\delta_1=4$ is the desired inflow to link 1. 

Let $(x(t),y(t))$ be the trajectory of the symmetric system \eqref{eq:9}--\eqref{eq:9-2} with initial condition $x_0=x(0)=(0,0,0)$ and $y_0=y(0)=(\bar{x}_1,\bar{x}_2,\bar{x}_3)=(6,4,2)$ where $g(x,y)$ is constructed as in the proof of Theorem \ref{thm:MM}. We have
\begin{alignat}{2}
  \label{eq:10}
  g(x_0,y_0)&=(\delta_1,0,0)^T \geq 0,\\
 g(y_0,x_0)&=(0,-d_2(\bar{x}_2),-d_3(\bar{x}_3))^T \leq 0,
\end{alignat}
that is, $(g(x_0,y_0),  g(y_0,x_0)^T)\succeq 0$ where we recall that $\succeq$ is the partial order induced by the orthant $\mathbb{R}^3_{\geq 0}\times \mathbb{R}^3_{\leq 0}$. It follows that $(x(t),y(t))$ is increasing with respect to $\succeq$ \cite[Ch. 3, Prop. 2.1]{Smith:2008fk}, and thus converges to an equilibrium $(x^*,y^*)$. By symmetry, we have that also $(y(t),x(t))$ is a trajectory of \eqref{eq:9}--\eqref{eq:9-2} converging to $(y^*,x^*)$. If $x^*=y^*=:\dense$, then we must have that $\dense$ is an equilibrium of the traffic flow dynamics defined by \eqref{eq:42}--\eqref{eq:42-2}. Furthermore, because the symmetric system is monotone with respect to $\succeq$, we must have that all trajectories of \eqref{eq:7}--\eqref{eq:7-2} converge to $(\dense,\dense)$, which in turn implies that $\dense$ is globally attractive for the traffic flow dynamics.

 In the top row of Figure \ref{fig:plots}, we plot the demand and supply curves given in \eqref{eq:7}--\eqref{eq:7-2}. We establish global convergence by verifying that an equilibrium $(x^*,y^*)$ of \eqref{eq:9}--\eqref{eq:9-2} satisfies $x^*=y^*$. %
The particular form of the supply and demand functions does not affect the qualitative behavior and \eqref{eq:7}--\eqref{eq:7-2} is chosen to be illustrative. The bottom row of Figure \ref{fig:plots} shows the trajectories $(x(t),y(t))$ and $(y(t),x(t))$ projected to the $x_l$$y_l$-plane for each $l$.
 \begin{figure}
   \centering
   \begin{tabular}{@{}c@{}@{}c@{}@{}c@{}}
\includegraphics[scale=.45,clip=true,trim=.05in 0.1in 0in 0.1in]{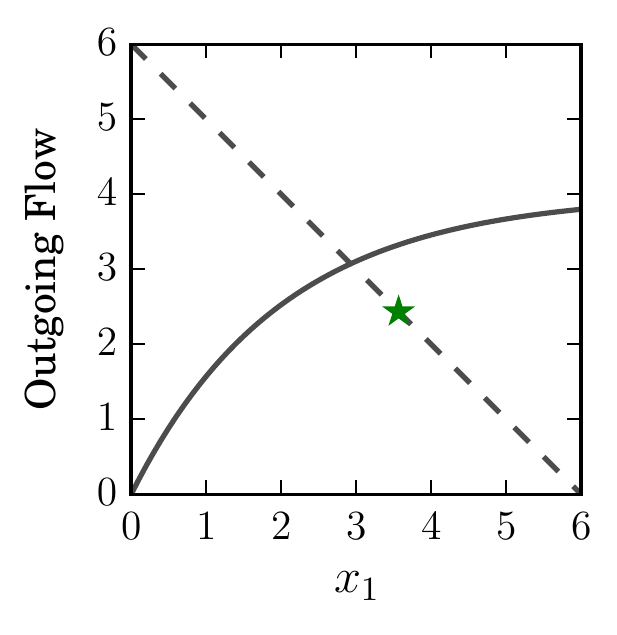}&     \includegraphics[scale=.45,clip=true,trim=.05in 0.1in 0in 0.1in]{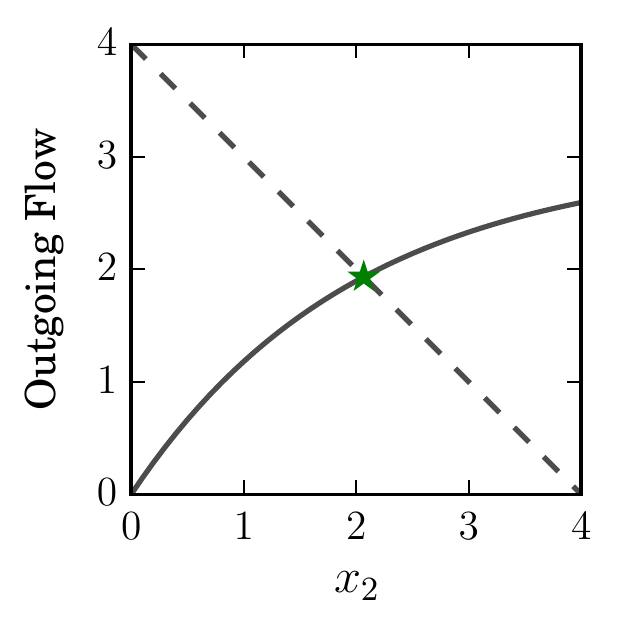}&  \includegraphics[scale=.45,clip=true,trim=.05in 0.1in 0in .1in]{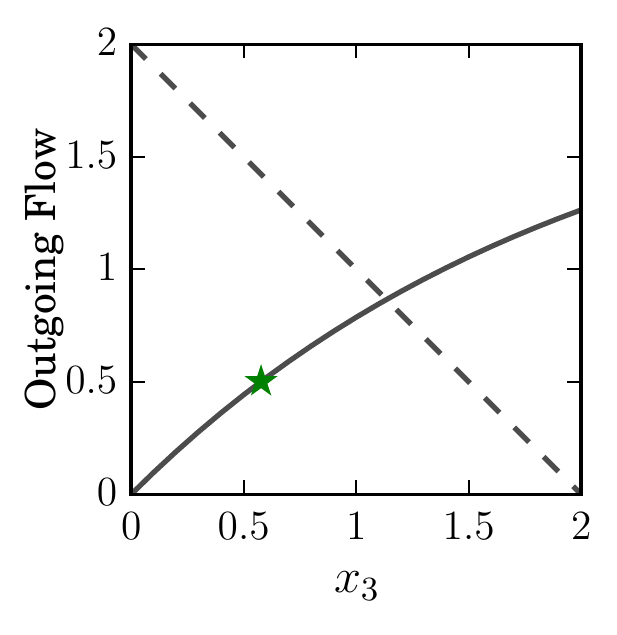}\\[-3pt]
 \includegraphics[scale=.45,clip=true,trim=.05in 0.1in 0in 0.1in]{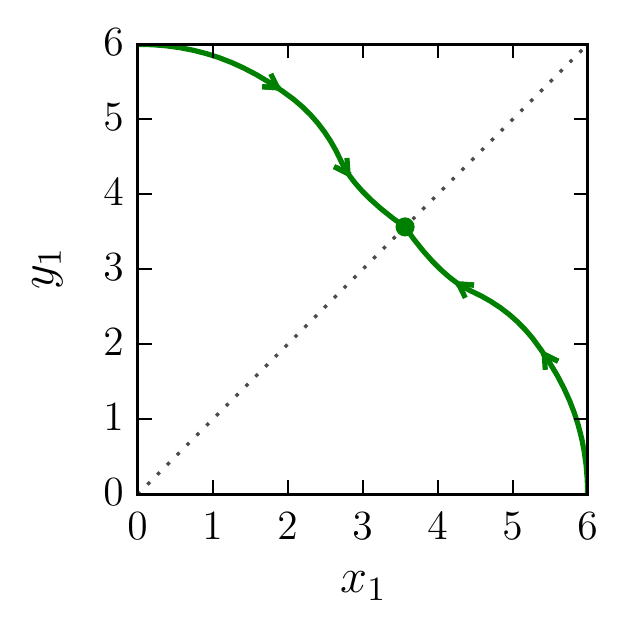}&\includegraphics[scale=.45,clip=true,trim=.05in 0.1in 0in .1in]{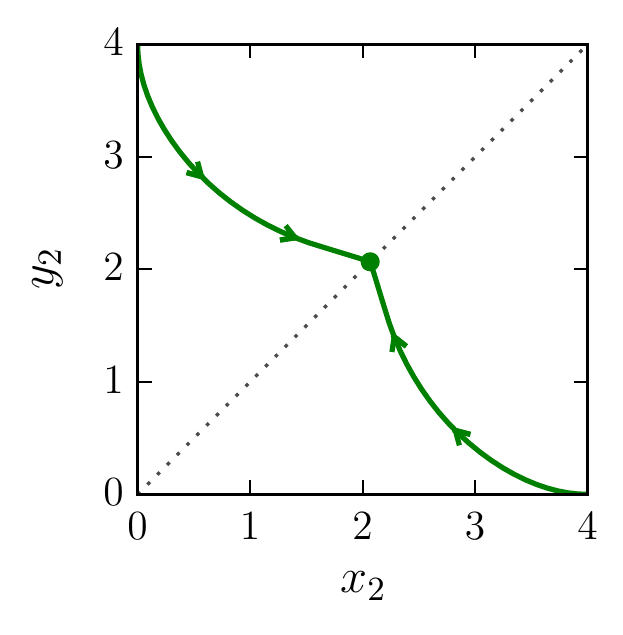}&\includegraphics[scale=.45,clip=true,trim=.05in 0.1in 0in 0.1in]{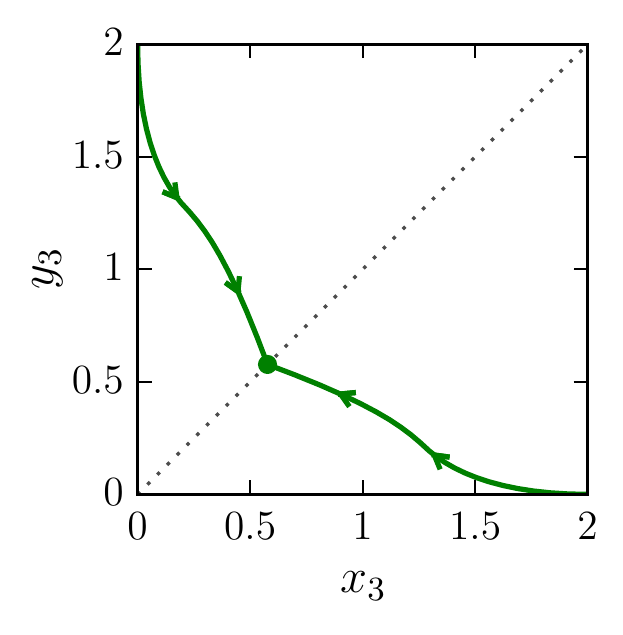}\\
\multicolumn{3}{c}{\includegraphics[scale=.45]{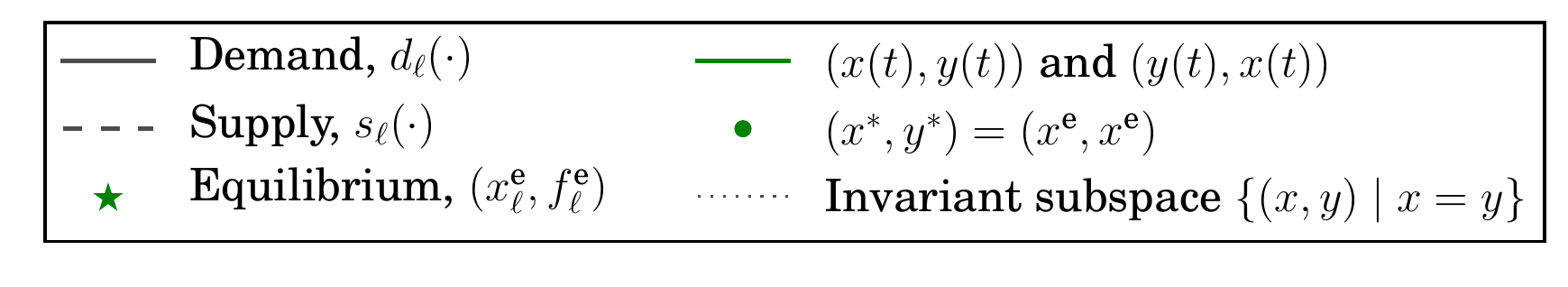}}
   \end{tabular}

   \caption{\noindent \textbf{(top)} Supply and demand curves for each link $l\in\{1,2,3\}$ of the diverging junction show in Figure \ref{fig:div} for the example in Section \ref{sec:example}. \textbf{(bottom)} Trajectories of the corresponding embedding system \eqref{eq:9}--\eqref{eq:9-2}. Mixed monotonicity ensures the embedding system is monotone with respect to the orthant $\mathbb{R}^n_{\geq 0}\times \mathbb{R}^n_{\leq 0}$. Convergence of the symmetric extreme trajectories $(x(t),y(t))$ and $(y(t),x(t))$ for $x(0)=0$, $y(0)=\bar{x}$ implies global convergence. The traffic flow dynamics are captured on a lower dimensional subspace of the embedding system, indicated by the dotted line.}
   \label{fig:plots}
 \end{figure}

\section{Conclusions}
\label{sec:conclusions}
We have proposed a general model for traffic flow networks that encompasses existing models and allows for a relaxed FIFO assumption at diverging junctions. We have shown that this general model is mixed monotone and have demonstrated how mixed monotonicity is used for system analysis such  as establishing global convergence. In future work, we will use mixed monotonicity to characterize traffic flow behavior for, \emph{e.g.}, cyclic networks, networks with excess demand, or networks with time-varying demand.

\section*{Acknowledgements}
The authors thank Gabriel Gomes and Giacomo Como for fruitful discussion regarding the limitations of non-FIFO and full FIFO traffic flow models.

\bibliographystyle{ieeetr}

\end{document}